\documentclass[11pt]{article}

\usepackage[letterpaper,margin=0.96in]{geometry}

\usepackage{bm, amsmath, amsthm, amssymb, amsfonts}

\usepackage{cite}
\usepackage{framed}
\usepackage[framemet hod=tikz]{mdframed}
\usepackage{appendix}
\usepackage{graphicx}
\usepackage{color}
\usepackage{varwidth}
\usepackage{wrapfig}
\usepackage{paralist}
\usepackage{booktabs}

\usepackage[textsize=tiny]{todonotes}

\usepackage{xspace}
\usepackage{xifthen}

\makeatletter
\renewcommand{\paragraph}{%
  \@startsection{paragraph}{4}%
  {\z@}{2.5ex \@plus 1ex \@minus .2ex}{-1em}%
  {\normalfont\normalsize\bfseries}%
}
\makeatother

\definecolor{darkgreen}{rgb}{0,0.5,0}
\definecolor{darkblue}{rgb}{0,0,0.8}
\usepackage{hyperref}
\hypersetup{
    unicode=false,          
    colorlinks=true,        
    linkcolor=darkblue,          
    citecolor=darkgreen,        
    filecolor=magenta,      
    urlcolor=brown           
}
\RequirePackage[]{silence}
\usepackage{algorithm}
\usepackage[noend]{algpseudocode}

\usepackage[nameinlink,capitalize]{cleveref}

\usepackage{thmtools}
\usepackage{thm-restate}

\newtheorem{theorem}{Theorem}[section]

\newtheorem{lemma}[theorem]{Lemma}

\usepackage{mathtools}

\newcommand{\calC}{\ensuremath{\mathcal{C}}}

\newcommand{\ignore}[1]{}

\algnewcommand\algorithmicswitch{\textbf{switch}}
\algnewcommand\algorithmiccase{\textbf{case}}

\algdef{SE}[SWITCH]{Switch}{EndSwitch}[1]{\algorithmicswitch\ #1\ \algorithmicdo}{\algorithmicend\ \algorithmicswitch}%
\algdef{SE}[CASE]{Case}{EndCase}[1]{\algorithmiccase\ #1}{\algorithmicend\ \algorithmiccase}%
\algtext*{EndSwitch}%
\algtext*{EndCase}%

\newcommand{\LOCAL}{\ensuremath{\mathsf{LOCAL}}\xspace}

\newcommand{\set}[1]{\left\{#1\right\}}

\newcommand{\Deltaline}{\bar{\Delta}}

\DeclareMathOperator{\polylog}{polylog}
\DeclareMathOperator{\polyloglog}{polyloglog}

\newcommand{\hide}[1]{}

\begin{document}


\title{\bf Distributed Edge Coloring in Time Quasi-Polylogarithmic in Delta}


\medskip
\author{
	 Alkida Balliu \\
	 University of Freiburg, Germany \\
	 alkida.balliu@cs.uni-freiburg.de
\and
	Fabian Kuhn \\
	University of Freiburg, Germany \\
	kuhn@cs.uni-freiburg.de
\and
	Dennis Olivetti \\
	University of Freiburg, Germany \\
	dennis.olivetti@cs.uni-freiburg.de
}

\date{}
\maketitle

\bigskip
 \begin{abstract}
  The problem of coloring the edges of an $n$-node graph of maximum degree $\Delta$ with $2\Delta - 1$ colors is one of the key symmetry breaking problems in the area of distributed graph algorithms. While there has been a lot of progress towards the understanding of this problem, the dependency of the running time on $\Delta$ has been a long-standing open question. Very recently, Kuhn [SODA '20] showed that the problem can be solved in time $2^{O(\sqrt{\log\Delta})}+O(\log^* n)$.
  
  In this paper, we study the edge coloring problem in the distributed \LOCAL model. We show that the $(\mathit{degree}+1)$-list edge coloring problem, and thus also the $(2\Delta-1)$-edge coloring problem, can be solved deterministically in time $\log^{O(\log\log\Delta)}\Delta + O(\log^* n)$. This is a significant improvement over the result of Kuhn [SODA '20].
 \end{abstract}

\section{Introduction \& Related Work}

An edge coloring of a graph $G=(V,E)$ is an assignment of colors to the edges $E$ of $G$ such that any two edges $e,e'\in E$ that share a common node $v\in V$ are assigned different colors. In the \emph{distributed edge coloring} problem, the graph $G$ models a network and the edges of $G$ have to be colored by using a distributed message passing algorithm on $G$. The typical goal is to color the edges with colors $1,\dots,2\Delta-1$, where $\Delta$ is the maximum degree of $G$. More specifically, the problem is most prominently studied in the so-called \LOCAL model, where the nodes of $G$ communicate in synchronous rounds and in each round, each node can exchange messages of arbitrary size with its neighbors in $G$ and perform some arbitrary internal computation~\cite{linial87,peleg00}. Note that each edge has at most $2\Delta-2$ conflicting (i.e., neighboring) edges and thus a coloring with $2\Delta-1$ colors can be obtained by a simple sequential greedy algorithm. Note also that the $(2\Delta-1)$-edge coloring problem is a special case of the $(\Delta+1)$-vertex coloring problem. Generally, distributed coloring is probably the most extensively studied problem in the area of distributed graph algorithms and certainly also one of the most widely studied problems in distributed computing in general. In the following, we only discuss the related work that is most relevant to the results in the present paper. For a relatively recent survey on distributed coloring in general, we refer to \cite{barenboimelkin_book}.

\paragraph{Distributed coloring as a function of the number of nodes.} It has been known since the 1980s that a $(2\Delta-1)$-edge coloring and also more generally a $(\Delta+1)$-vertex coloring can be computed by a $O(\log n)$-round randomized algorithm in the \LOCAL model~\cite{alon86,luby86,linial87}. Whether the problems can be solve similarly efficiently by a deterministic distributed algorithm has been a highly important open problem in the area for 30 years~\cite{linial87,barenboimelkin_book,derandomization,stoc17_complexity}. It was recently shown that the $(2\Delta-1)$-edge coloring problem can also be solved deterministically in $\polylog n$ rounds~\cite{FischerGK17,Harris18}. In a very recent breakthrough, Rozho\v{n} and Ghaffari~\cite{polylogdecomp} showed that in the \LOCAL model, every (locally checkable) problem with a $\polylog n$-round randomized solution can also be solved in $\polylog n$ time deterministically, and hence also the $(\Delta+1)$-vertex coloring problem can be solved deterministically in $\polylog n$ rounds. In combination with some recent advanced randomized distributed coloring algorithms, these results imply that both the $(2\Delta-1)$-edge coloring problem and also the more general $(\Delta+1)$-vertex coloring problem can be solved in time $\polyloglog n$~\cite{ElkinPS15,chang18_coloring}.

\paragraph{Distributed coloring as a function of the maximum degree.} The general objective when solving a distributed problem in the \LOCAL model is to understand to what extent it is sufficient for a node $v$ to only gather local information in order to compute $v$'s part of the solution. This question is most interesting in very large networks, where the maximum degree $\Delta$ might be much smaller than $n$ or even  independent of $n$. It is therefore natural to ask about the complexity of a problem as a function of $\Delta$, rather than as a function of $n$. This might also allow to determine bounds on the locality that only depend on local properties and are in particular (almost) independent of the size of the network. By the classic lower bound of Linial~\cite{linial87}, it is known that $\Omega(\log^* n)$ rounds are needed even to $3$-color the nodes (or edges) of an $n$-node cycle. The typical objective therefore is to find the best possible complexity of the form $O(f(\Delta)+\log^* n)$. Linial~\cite{linial87} showed that in $O(\log^* n)$ rounds, it is possible to (deterministically) compute an $O(\Delta^2)$-coloring. In the case of the distributed coloring problem, the question of finding the best complexity of the form $O(f(\Delta)+\log^* n)$ is thus also closely related to the question of reducing the number of colors of a given initial coloring in time that only depends on the number of initial colors and is independent of $n$.

As it is straightforward to reduce the number of colors of a given coloring by $1$ in a single round, the $O(\log^* n)$-round $O(\Delta^2)$-coloring algorithm of \cite{linial87} immediately implies an $O(\Delta^2 + \log^* n)$-round algorithm for computing a $(\Delta+1)$-coloring (or a $(2\Delta-1)$-edge coloring). With a slightly more clever color reduction scheme, the time complexity can be improved to $O(\Delta\log\Delta + \log^* n)$~\cite{szegedy93, Kuhn2006On}. By introducing the concept of distributed defective colorings and a divide-and-conquer approach to the distributed coloring problem, Barenboim and Elkin~\cite{barenboim09}, as well as Kuhn~\cite{spaa09} showed that the time for computing a $(\Delta+1)$-vertex coloring can be improved to $O(\Delta+\log^* n)$. For the $(2\Delta-1)$-edge coloring problem, the same time complexity was already known from an earlier paper by Panconesi and Rizzi~\cite{PanconesiR01}. In \cite{barenboim15}, Barenboim managed to develop the first algorithm with a time complexity that is sublinear in $\Delta$ by giving an $O(\Delta^{3/4}\log\Delta + \log^* n)$-round $(\Delta+1)$-vertex coloring algorithm. The most important novel idea of \cite{barenboim15} was to consider the more general list coloring problem. In a $(\Delta+1)$-list vertex coloring, each node is initially given a list of $\Delta+1$ arbitrary colors and in the end, each node must be colored with one of the colors from its list. Being able to solve list coloring in particular allows to extend an initial partial coloring of a graph to a full coloring of the graph. The ideas of \cite{barenboim15} were further developed by Fraigniaud, Heinrich, and Kosowski~\cite{fraigniaud16}, who obtain a $(\Delta+1)$-vertex coloring algorithm that runs in $O(\sqrt{\Delta}\cdot\polylog\Delta + \log^* n)$ rounds in the \LOCAL model. A small improvement of \cite{barenboim18} yields an algorithm with time complexity $O(\sqrt{\Delta\log\Delta}\cdot\log^*\Delta + \log^* n)$. This is the fastest known algorithm for the $(\Delta+1)$-vertex coloring problem and thus also the fastest known algorithm that works for both the edge and the vertex coloring problem. For the edge coloring problem, very recently Kuhn~\cite{soda20_coloring} showed that the algorithm of \cite{fraigniaud16,barenboim18} can be improved significantly and that it is possible to compute a $(2\Delta-1)$-edge coloring in time $2^{O(\sqrt{\log\Delta})}+O(\log^* n)$. The main result of this paper is another substantial improvement for the edge coloring problem and we show that a $(2\Delta-1)$-edge coloring can be computed in time $\log^{O(\log\log\Delta)}\Delta + O(\log^* n)$ rounds and thus in time quasi-polylogarithmic in $\Delta$. We discuss our contribution in more detail below.

We note that the distributed edge coloring problem has been known to have an easier structure than the vertex coloring problem. The probelm has been widely studied, both in the deterministic and the randomized setting, and there has been plenty of progress in the past years (in addition to the papers already mentioned, see, e.g., \cite{panconesi1997randomized, dubhashi1998near, CzygrinowHK01, ghaffari17, derandomization,ChangHLPU18, stoc18_edgecoloring, SuVu19}). While the most classic version of the problem asks to color the edge of a graph with $2\Delta-1$ colors, in the literature, there are algorithms that use color palettes of various sizes. Notice that while $(2\Delta-1)$-coloring can be solved in time $O(f(\Delta)+\log^*n)$, for the $(2\Delta - 2)$-edge coloring problem, a lower bound of $\Omega(\log n)$ is known even for bounded-degree graphs~\cite{LLL_lower}.

\paragraph{Our contribution.}
We make progress in our understanding of the complexity of the edge coloring problem in the \LOCAL model of distributed computation. More specifically, we give a novel algorithm to the $(\deg(e)+1)$-list edge coloring problem, which is defined as follows. For an edge $e=\set{u,v}\in E$ of a graph $G=(V,E)$, we define $\deg(e):=\deg(u)+\deg(v)-2$ to be the degree of $e$ (i.e., $\deg(e)$ is equal to the number of neighboring edges of $e$). In an instance of the $(\deg(e)+1)$-list edge coloring problem on $G=(V,E)$, each edge $e\in E$ of $G$ is given a list $L_e$ of $\deg(e)+1$ colors and the objective is to compute a valid edge coloring, where each edge $e$ is colored with a color from its list $L_e$. Throughout the paper, we assume that all the lists $L_e$ consist of colors from $\set{1,\dots,\Delta^c}$ for some constant $c>0$. For this setting, we
prove the following main theorem.

\begin{framed}
  \begin{theorem}\label{thm:main}
    There is a deterministic distributed algorithm that solves $(\deg(e) +1)$-list edge coloring problem in $\log^{O(\log \log \Delta)} \Delta + O(\log^*n)$ rounds in the \LOCAL model.
  \end{theorem}
\end{framed}

The $(\deg(e)+1)$-list edge coloring problem is clearly a generalization of the $(2\Delta-1)$-edge coloring problem. As a corollary, we therefore immediately get that also the  $(2\Delta-1)$-edge coloring problem can be solved in quasi-polylogarithmic in $\Delta$ deterministic rounds in the \LOCAL model.

\section{Preliminaries}
\subsection{Definitions}
Let $G=(V,E)$ be an undirected graph. We denote as $\Delta$ the maximum degree of the nodes of $G$, and as $\Deltaline$ the maximum degree in the line graph of $G$. Clearly, $\Deltaline \le 2\Delta-2$. We use $\deg(v)$ to denote the degree of a node $v \in V$, and we use $\deg(e)$ to denote the degree of $e$ in the line graph of $G$. 

We will further assume that for an integer $p\geq 1$, $H_p:=\sum_{i=1}^p 1/i$ denotes the
$p^{\mathit{th}}$ harmonic number and we assume that $\log x$ denotes
the logarithm to the base $2$.

\paragraph{List edge coloring}
Given a graph $G=(V,E)$ where a list $L_e$  known to both endpoints of $e$ is assigned to every edge $e\in E$, the (distributed) list edge coloring problem requires the nodes to color each incident edge with an element in $L_e$, satisfying the constraint that if two edges are incident to the same node, their assigned color must be different. The $(\deg(e)+1)$-list edge coloring problem is defined to be the special case that satisfies $|L_e| > \deg(e)$ for all $e \in E$.

\paragraph{Edge coloring}
The edge coloring problem is a special case of list edge coloring, where all the lists assigned to the edges are the same. The $(2\Delta-1)$-edge coloring problem is the case where all lists have size $2\Delta-1$, that is the amount required to make the problem greedily solvable in the centralized setting. Notice that a solution for the  $(\deg(e)+1)$-list edge coloring problem implies a solution for the $(2\Delta-1)$-edge coloring problem.

\paragraph{Defective coloring}
The $d$-defective $c$-coloring problem requires to color nodes of a graph $G=(V,E)$ with $c$ colors, such that the subgraphs induced by nodes of the same color have degree at most $d$. In the defective edge coloring case, the edges must be colored with $c$ colors, and $d$ is an upper bound on the degree of the line graph induced by edges of the same color. We sometimes require that the defect of each edge $e\in E$ depends on the degree $\deg(e)$ of $e$. The definition of a defective edge coloring is then generalized in the natural way: a $f(e)$-defective edge coloring is a coloring of the edges, where each edge $e\in E$ has at most $f(e)$ neighboring edges of the same color.

\subsection{LOCAL Model}
We consider the standard \LOCAL model of distributed computing. In this model, we have a communication network that can be represented as a graph $G=(V,E)$. Nodes represent entities that can perform computation, while edges represent communication links. This model is synchronous. All nodes start the computation at the same time, and then the computation proceeds in rounds. In each round, nodes can exchange a message with each neighbor. At the beginning nodes know the size $n = |V|$ of the graph, the maximum degree $\Delta$ of the graph, and their ID, that is a value in $\{1,\dotsc,n^{O(1)}\}$ different from the ones assigned to all the other nodes of the graph. At the end nodes must output their part of the solution. For example, for the edge coloring problem each node must output the color of each of its incident edges. If two nodes are incident on the same edge, they must output a consistent value for that edge. We say that an algorithm runs in $T$ rounds if all nodes output their part of the solution within $T$ rounds of communication. We consider deterministic algorithms: nodes do not have access to random bits. 

\section{Key Ideas}
In this paper, we show that the $(\deg(e)+1)$-list edge coloring can be solved in $\log^{O(\log \log \Delta)} \Delta + O(\log^*n))$ deterministic rounds in the \LOCAL model. Since the $(2\Delta - 1)$-edge coloring problem is a special case of the list edge coloring problem, we directly get as a corollary that the $(2\Delta - 1)$-edge coloring problem can be solved in deterministic quasi-polylogarithmic in $\Delta$ rounds in the \LOCAL model.
In order to show our claim, we use a technique already presented in \cite{soda20_coloring} (and which is based on ideas from \cite{barenboim15,fraigniaud16}). We show a procedure that solves more relaxed \emph{``easy''} instances of the list coloring problem, and then we use this procedure to solve the ``hard'' case in a black box manner.

\paragraph{Relaxed list edge coloring.}
Let $P(\Deltaline,1,C)$ be the family of $(\deg(e)+1)$-list edge coloring problems on graphs with maximum edge-degree $\Deltaline$, where the color palette has size $C$. In the relaxed version of the list coloring problem we require that, for each edge $e$, the size of $L_e$ is large. More precisely, we define as $P(\Deltaline,S,C)$ the family of list edge coloring problems on graphs with maximum edge-degree $\Deltaline$, where the color palette has size $C$, and the lists have slack $S$, that is, the list of each edge $e$ has size strictly greater than $S\cdot \deg(e)$. Also, let $T(\Deltaline,S, C)$ be the time required to solve $P(\Deltaline,S,C)$.  We show that we can reduce a single list edge coloring instance to many relaxed list edge coloring instances.
\paragraph{From slack \boldmath $S$ to slack $1$.}
Suppose we are given an initial $X$-edge coloring. We show that, if we have an algorithm $A(\Deltaline,S, C)$ that solves $P(\Deltaline,S,C)$ in time $T(\Deltaline,S, C)$, then we can solve $P(\Deltaline,1,C)$ in time
\[
T(\Deltaline, 1, C)\leq O(S^2\cdot\log \Deltaline)\cdot T(\Deltaline, S, C) + O(\log \Deltaline \log^*X).
\]
In other words, we show that in order to solve a single instance with no slack, we can sequentially solve roughly $S^2$ instances with higher slack $S$. The idea is that, while we can not make the lists larger, we can try to make the degree smaller. We start by computing a defective edge coloring, that basically decomposes our graph in many subgraphs of smaller degree. Then we process our graphs sequentially, and we color each graph with the algorithm that requires higher slack. The idea is that each time we apply this algorithm, we put in the list of the edges only the colors that are still unused, and if a list gets too small, then the edge does not participate in the process. We finally recurse on the uncolored edges, and we show that by choosing parameters wisely, this does not take too much time.

\paragraph{How to solve \boldmath $P(\Deltaline,S,C)$.}
In \cite{soda20_coloring}, Kuhn presents a technique called list color space reduction. In general, this technique splits the color space into many independent subspaces, and tries to assign a subspace to each edge. By doing so, it is possible to independently recurse on each graph induced by edges with the same assigned subspace. The main contribution of our work is a faster list color space reduction algorithm for the case of edge coloring.
Suppose we are given a list edge coloring instance with a color palette of size $C$, where each edge $e$ has its list $L_e$. We split the color palette roughly into $p$ parts, $C_1,\dotsc, C_p$, each of size at most $C/p$, and assign an index $i\in\{1,\dotsc,p\}$ to edges. Then, each edge $e$ having index $i$ updates its list as $L_e=L_e \cap C_i$. In this way, we divide the list edge coloring instance into $p$ independent list edge coloring instances, each with a color palette of size $C/p$. Then, each of the $p$ problem instances can be solved in parallel, by using the same algorithm in the subgraph induced by edges that have the same index. We need to be careful: for each edge $e$, the size of the list $L_e$ must not reduce too much compared to the degree that $e$ has in the subgraph induced by edges that have the same index as $e$. We show that, by carefully creating lists of indexes, and by using a fast defective edge coloring algorithm, it is possible to do this operation fast.

\paragraph{Putting things together}
By reducing a list coloring problem to many relaxed list edge coloring problems, we basically get $T(\cdot, 1, \cdot)$ as a function of $T(\cdot, S, \cdot)$. Then, by solving $P(\cdot, S, \cdot)$, we will essentially show how to express $T(\cdot, S, \cdot)$ as a function of $T(\cdot, S', \cdot)$ and $T(\cdot, 1, \cdot)$, where $S'\le S$. Having all this, we can find a closed formula for $T(\cdot, 1, \cdot)$ and obtain our quasi-polylogarithmic in $\Delta$ time complexity.

\section{Edge Coloring}
In this paper, we will not only prove that the $(2\Delta-1)$-edge coloring problem can be solved in quasi-polylogarithmic in $\Delta$ time, but we will prove a stronger statement. In fact, we will show that, if each edge $e$ is provided with a list of $\deg(e)+1$ colors, it is possible to assign to each edge a color from its list, such that neighboring edges have different colors, and this can be done in time quasi-polylogarithmic in $\Delta$. More formally, we will prove the following theorem.
\begin{theorem}\label{thm:listedgecoloring}
	The $(\deg(e) +1)$-list edge coloring problem can be solved in $\log^{O(\log \log \Delta)} \Delta + O(\log^*n)$ deterministic rounds in the \LOCAL model.
\end{theorem}

We start by giving some definitions and simple observations. Given a graph $G=(V,E)$, let $\Delta$ be the maximum degree  of $G$ and $\Deltaline$ be the maximum degree of the line graph of $G$. Note that, as long as $\Deltaline > 0$, we have $\Deltaline = \Theta(\Delta)$. We parametrize the list edge coloring problem by three parameters. For integers $\Deltaline\geq 0$, $C\geq 1$ and a slack parameter $S\geq 1$, $P(\Deltaline,S,C)$ is the list edge coloring problem where the maximum edge degree is $\Deltaline$, the color palette has size $C$, and the lists have slack $S$, that is, the list of each edge $e$ has size strictly greater than $S \cdot \deg(e)$. Note that for $S=1$, the problem corresponds to a $(\deg(e)+1)$-list edge coloring problem. We define $T(\Deltaline, S, C)$ to be the time required to solve $P(\Deltaline,S,C)$. Similarly, $T(\Deltaline, 1, C)$ is the time required to solve $P(\Deltaline,1,C)$. In order to prove \Cref{thm:listedgecoloring}, we will provide an upper bound for $T(\Deltaline,1,C)$.

Let us now make some observations. First, we clearly have $T(\Deltaline',S,C)\leq T(\Deltaline, S,C)$ for all $\Deltaline$, $S$, and $C$, and every $\Deltaline'\leq \Deltaline$.
In addition, note that even if we allow graphs where the line graph has maximum degree at most
$\Deltaline$, the other two parameters $S$ and $C$ might further restrict
the maximum degree. If an edge has degree $d$, its list has to be of size
larger than $d\cdot S$ and therefore $C$ has to be larger than $d\cdot
S$. The maximum line graph degree can therefore be at most the largest integer
that is smaller than $C/S$, i.e., $\lceil C/S\rceil-1$. For every
$\Deltaline$, $S$, and $C$, we therefore have
\[
T(\Deltaline, S, C) = T\left(\min\set{\Deltaline, \left\lceil \frac{C}{S}\right\rceil-1}, S, C\right).
\]
Note that this in particular implies that if $S\geq C$, the maximum
possible degree is $0$ (while the minimum list length is always at
least $1$) and thus, we have solved the problem. Further, if $\Deltaline$
is a constant and some initial $X$-edge coloring of $G$ is given, the
time to solve any edge list coloring instance is $O(\log^* X)$~(see, e.g., \cite{linial87,PanconesiR01}), i.e.,
\[
T\big(O(1), S, C\big) = O(\log^* X).
\]

We will start by using a technique already presented in \cite{soda20_coloring}, where it has been proved that we can reduce a hard list coloring instance to many easier list coloring instances. This technique will allow us to provide an algorithm that requires high slack on the list size, while still allowing us to solve the hardest case where the slack has value $1$. In particular, we will prove the following lemma.
\begin{lemma}\label{lemma:hardtoeasy}
	For any $\beta>1$, any $(\deg(e)+1)$-list edge coloring instance can
	be reduced to consecutively solving $O(\beta^2\log\Delta)$ list
	edge coloring instances with slack $\beta$. More precisely, if an initial edge coloring
	with $X$ colors is given, for all $\beta> 1$,
	\begin{align*}
	T(\Deltaline, 1, C) &\leq O(\log^*X) + O(\beta^2)\cdot
	T\left(\frac{\Deltaline}{2\beta}, \beta, C\right) + T\left(\frac{\Deltaline}{2}, 1, C\right) \\
	&\leq
	O(\beta^2\cdot\log \Deltaline)\cdot T(\Deltaline, \beta, C) + O(\log \Deltaline \log^*X).
	\end{align*}
\end{lemma}
Thus, we will provide an algorithm that solves $P(\Deltaline,S, C)$ for a value $S$ of our choice, and by applying \Cref{lemma:hardtoeasy} we will automatically get an algorithm that solves the case where $S = 1$.

In order to prove \Cref{lemma:hardtoeasy}, we will further decompose the problem into many easier problems. Essentially, we will provide a way to split the color space of size $C$, and assign to each edge a smaller color space of size $C / p$, for some parameter $p$ of our choice. In this way, we can recursively solve the problem separately on each subgraph induced by edges that have the same subspace. More formally, we will prove the following lemma.
\begin{lemma}\label{lemma:colorspacereduction}
	Let $p\in[2, C]$ be an integer parameter. If $\Deltaline\geq 1$, $C\geq 2$, and
	$S\geq 24\cdot H_{2p}\cdot \log p$, we can express
	$T(\Deltaline, S, C)$ recursively as
	\[
	T(\Deltaline, S, C) \leq (\log p)\cdot(1+T(2p-1, 1, 2p)) + 
	T\left(\Deltaline, \frac{S}{24\cdot H_{2p}\cdot \log p}, \frac{C}{p}\right).
	\]
\end{lemma}

In \Cref{subsec:hardtoeasy} we will prove \Cref{lemma:hardtoeasy}. Then, in \Cref{subsec:colorspacereduction} we will prove \Cref{lemma:colorspacereduction}. Finally, in \Cref{subsec:theorem} we will use these lemmas to prove our main result, \Cref{thm:listedgecoloring}.

\subsection{Proof of \Cref{lemma:hardtoeasy}}\label{subsec:hardtoeasy}
While a similar result has been already proved in a more general form in \cite{soda20_coloring}, in order to make our result self-contained, we provide a proof that suits our specific needs.

The high level idea is the following. We want to solve the case where the slack is $1$, that is, the lists may contain just one element more than what is strictly necessary. In other words, if an edge $e$ has degree $\deg(e)$, the list $L_e$ may contain just $\deg(e) +1$ colors, that is the minimum amount that makes the problem always solvable in a greedy manner. While we cannot make the lists larger, we can try to make the degrees smaller. In particular, we will split our graph into many subgraphs of smaller degree, and we will iteratively color these subgraphs. Each time, we will remove from the lists of the edges the colors used by their neighbors in the previous steps, and only edges having a large remaining list will participate in the coloring phase. At the end, we will take care of edges that could not participate in the coloring phase, by recursively applying the same technique on the subgraph induced by those edges. 

In the following, we will consider two different colorings: the color $g(e)$ identifies a temporary color assigned to $e$, and we will consider subgraphs induced by edges of the same color, while the color $c(e)$ will be the final result of the algorithm (initially $c(e) = \bot$ for all $e$, that is, it is unassigned). More formally, the algorithm proceeds as follows:
\begin{enumerate}
	\item If $\Deltaline = O(1)$, solve the problem into $O(\log^* X)$ with any standard list coloring algorithm~(e.g., \cite{linial87,PanconesiR01}). Otherwise,
	\item Compute a $\frac{\deg(e)}{2 \beta}$-defective edge coloring with $O(\beta^2)$ colors. Let $g(e)$ be the color assigned to edge $e$.
	\item Iterate over the $O(\beta^2)$ color classes. For each color $i$, do the following:
	\begin{enumerate}
		\item Each edge $e$ such that $g(e) = i$ removes from its list the colors $c(e')$ used by each neighboring edge $e'$ (if different from $\bot$).
		\item If the new list $L_e$ satisfies $|L_e| > \deg(e)/2$ then the edge is marked as \emph{active}.
		\item Apply the algorithm solving $P(\frac{\Deltaline}{2 \beta}, \beta, C)$ on the subgraph induced by edges $e$ satisfying $g(e) = i$ that are marked \emph{active}.
	\end{enumerate}
	\item Recurse on the subgraph induced by uncolored edges (the ones for which  $c(e) = \bot$).
\end{enumerate}
An example of the execution of this algorithm is shown in Figures \ref{fig:list-edge-col}, \ref{fig:first-color}, \ref{fig:second-color}, and \ref{fig:third-color}. 
We will first provide a simple algorithm that computes the required defective coloring. Then, we will prove that the lists are large enough to apply the algorithm for $P(\Deltaline, \beta, C)$ (the slack must be at least $\beta$). Finally, we will prove an upper bound on the degree of the subgraph induced by uncolored edges, that will allow us to bound the number of recursive steps. This will conclude the proof of \Cref{lemma:hardtoeasy}.

\begin{figure}
	\centering
	\includegraphics[width=0.4\textwidth]{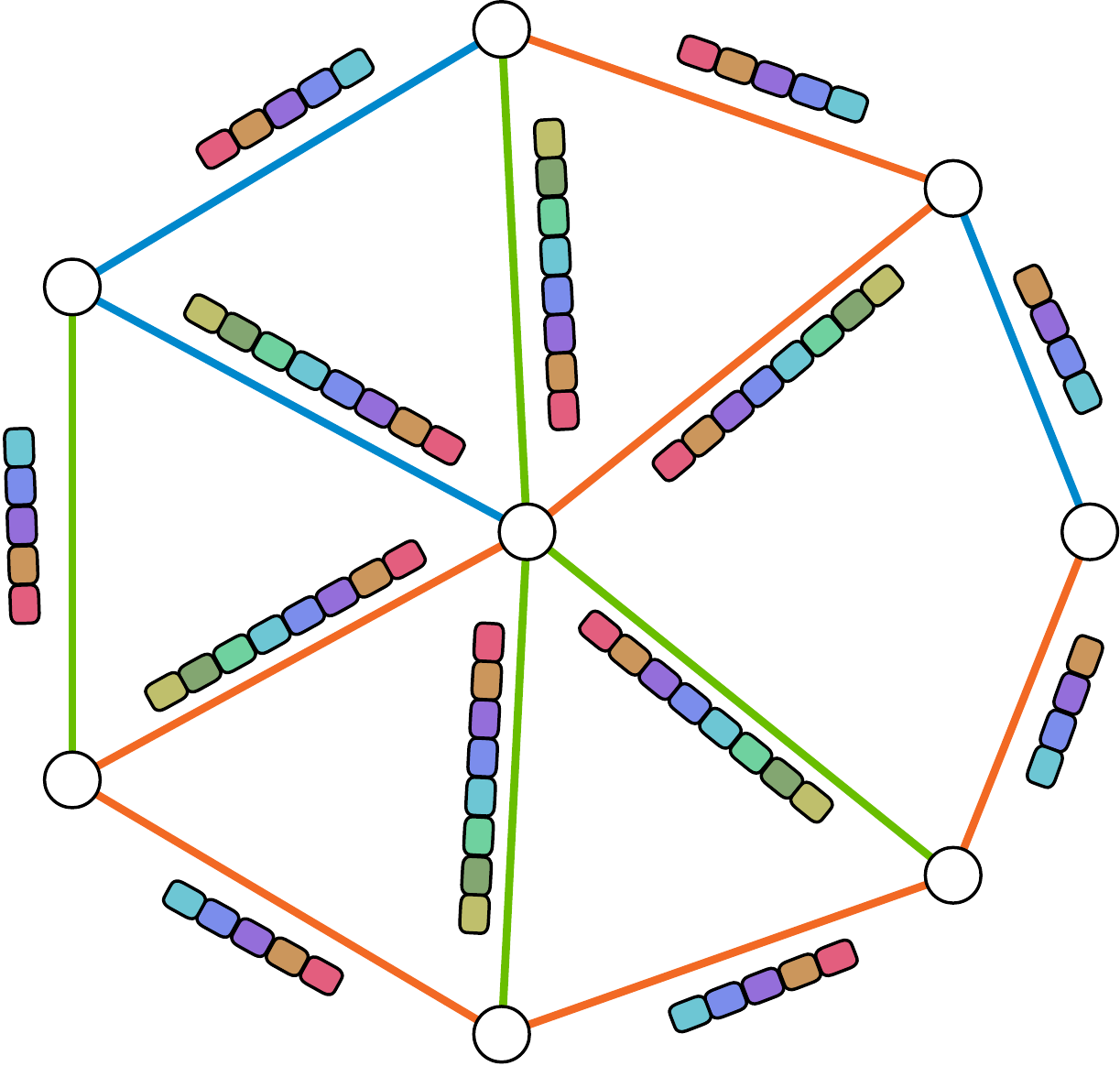}
	\caption{An example of list coloring instance where a defective edge coloring has been computed.  The colors over the edges represent their lists, while the color of each edge $e$ represents $g(e)$.}
	\label{fig:list-edge-col}
\end{figure}

\begin{figure}
	\centering
	\begin{minipage}{.5\textwidth}
		\centering
		\includegraphics[width=.65\linewidth]{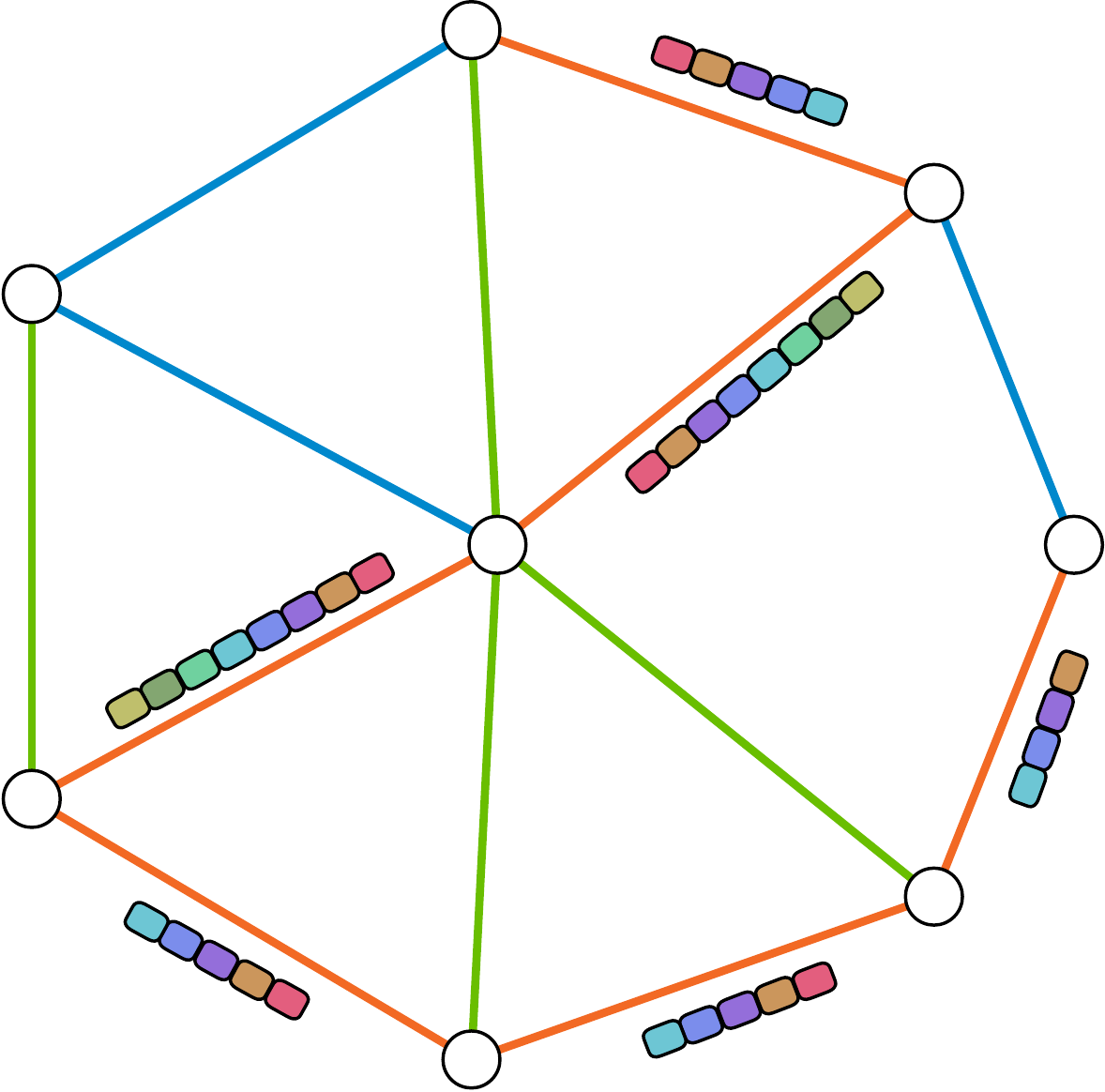}
	\end{minipage}%
	\begin{minipage}{.5\textwidth}
		\centering
		\includegraphics[width=.65\linewidth]{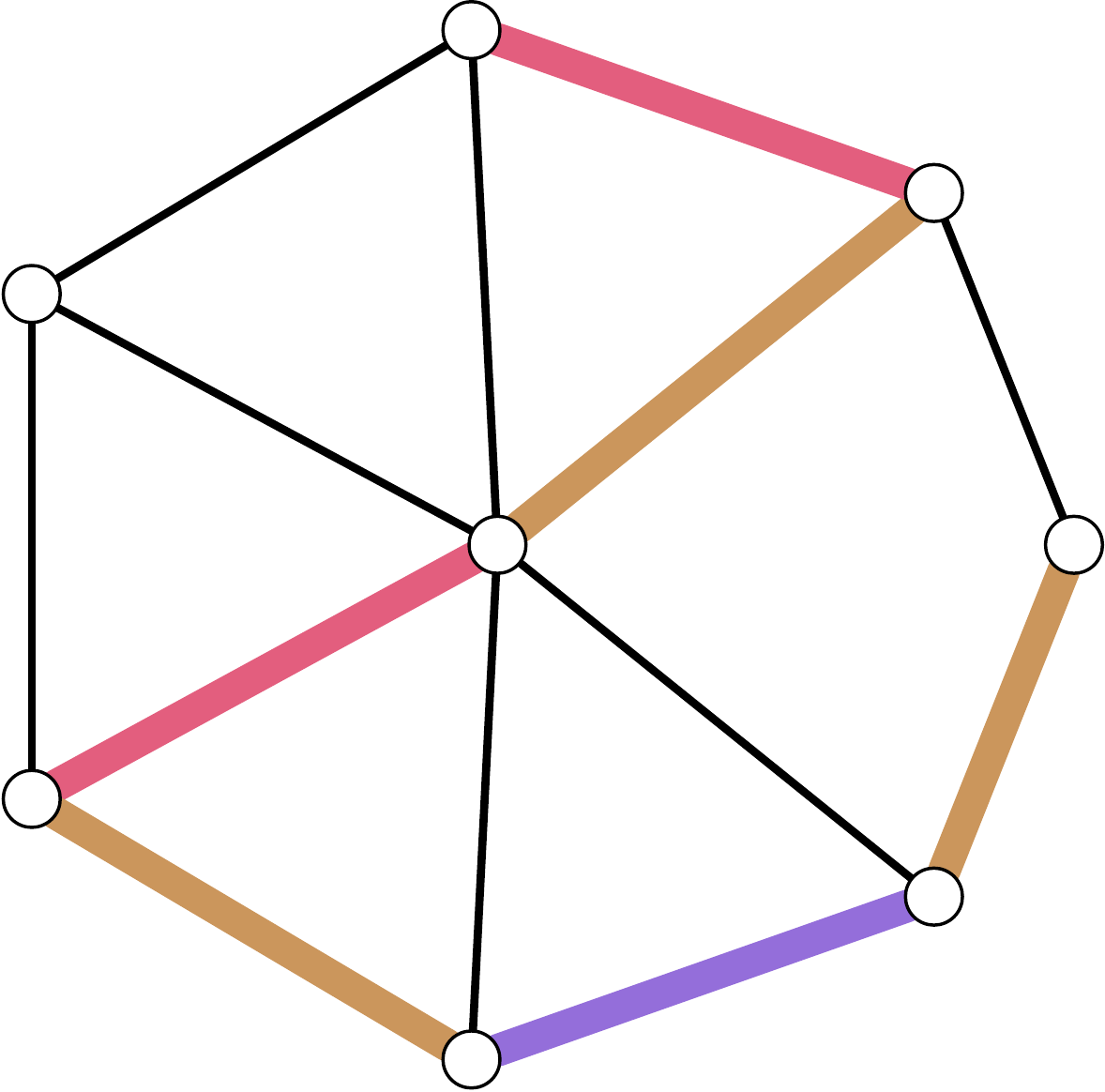}
	\end{minipage}
	\caption{The algorithm that requires slack is applied on the subgraph induced by red edges. The result is shown on the right, where bold edges have been colored.}
	\label{fig:first-color}
\end{figure}

\begin{figure}
	\centering
	\begin{minipage}{.5\textwidth}
		\centering
		\includegraphics[width=.65\linewidth]{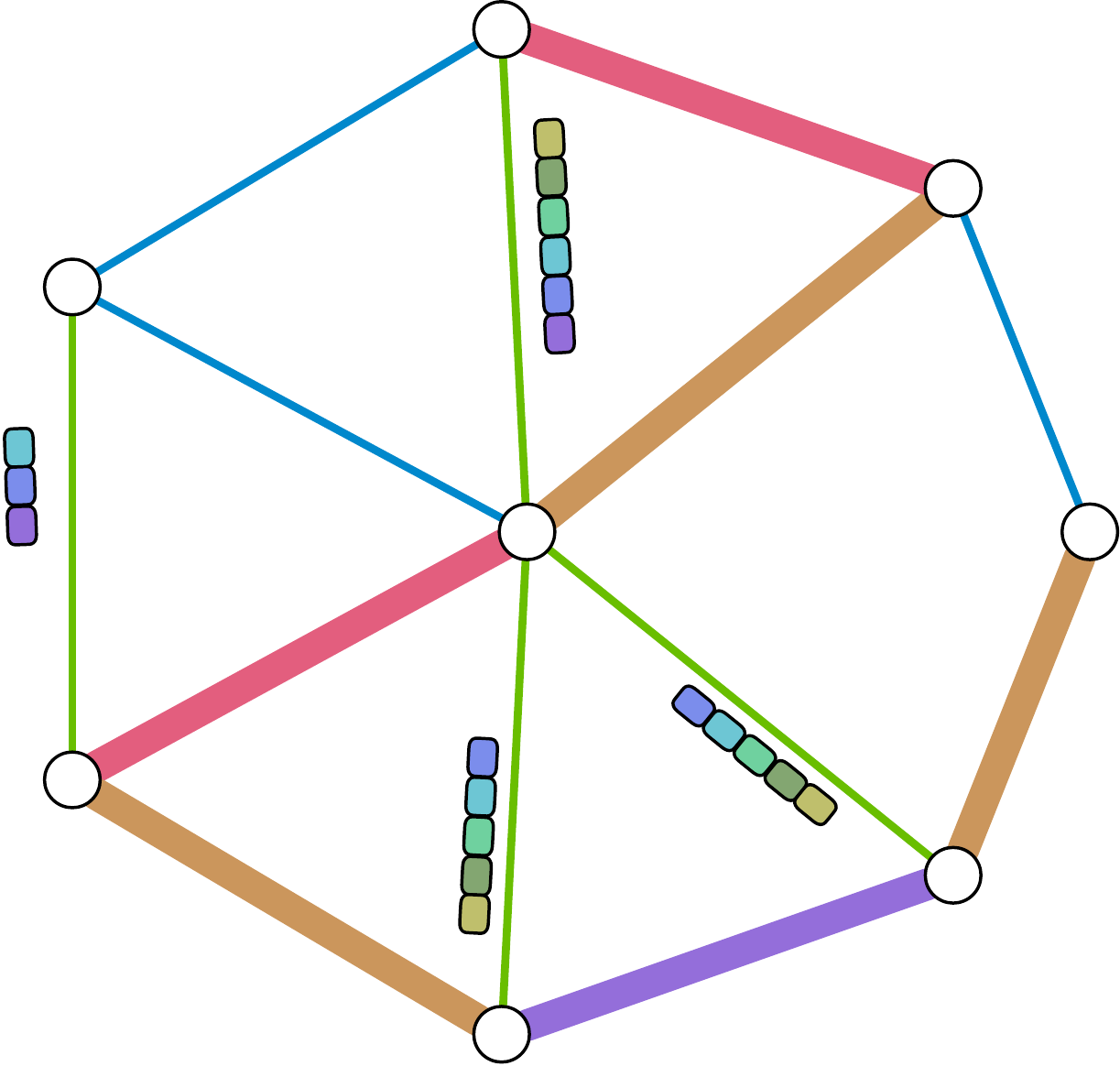}
	\end{minipage}%
	\begin{minipage}{.5\textwidth}
		\centering
		\includegraphics[width=.65\linewidth]{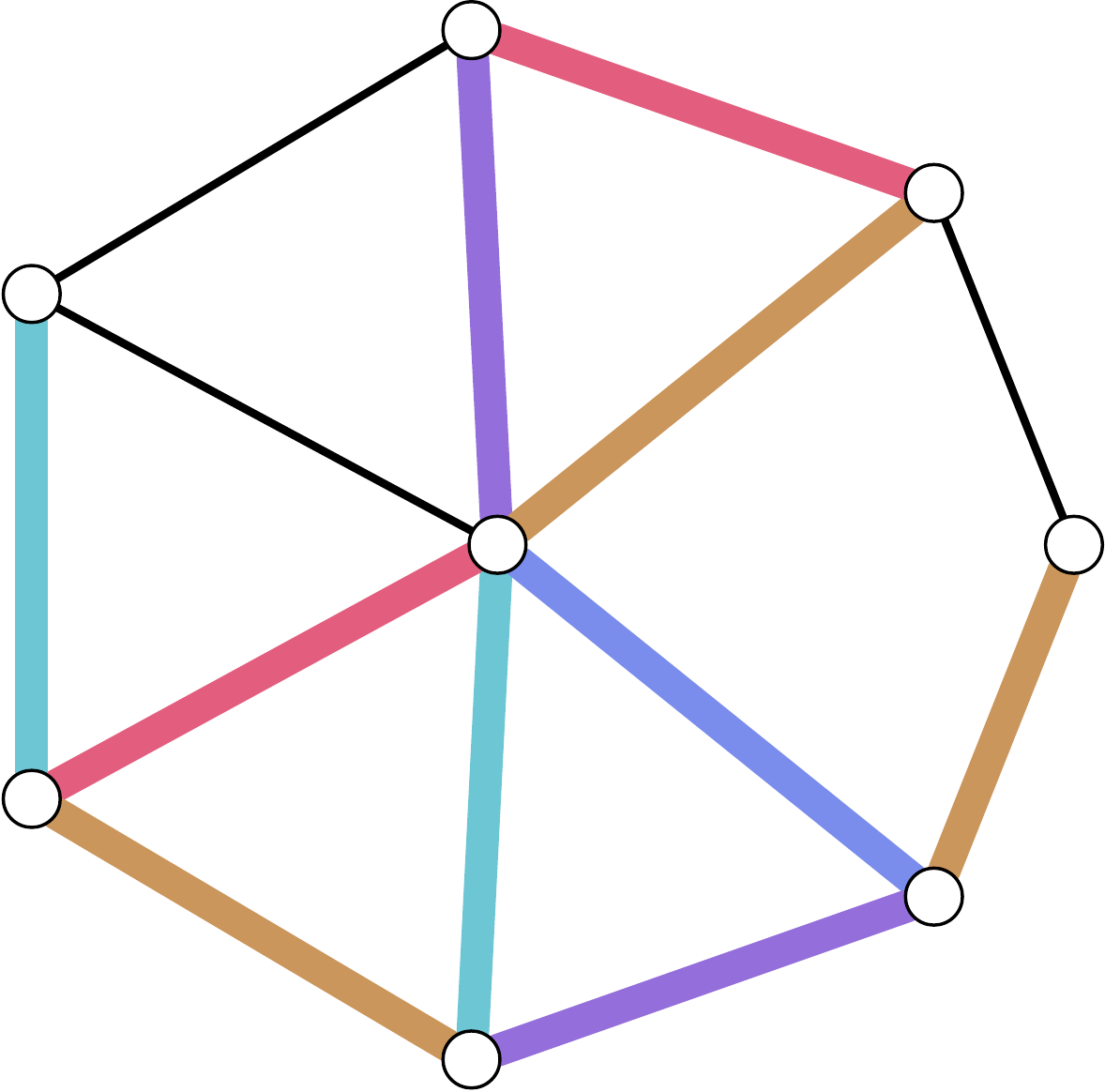}
	\end{minipage}
	\caption{Each green edge has a list of size strictly greater than $\deg(e)/2$, hence all of them are active while applying the coloring algorithm. The result is shown on the right.}
	\label{fig:second-color}
\end{figure}

\begin{figure}
	\centering
	\begin{minipage}{.5\textwidth}
		\centering
		\includegraphics[width=.65\linewidth]{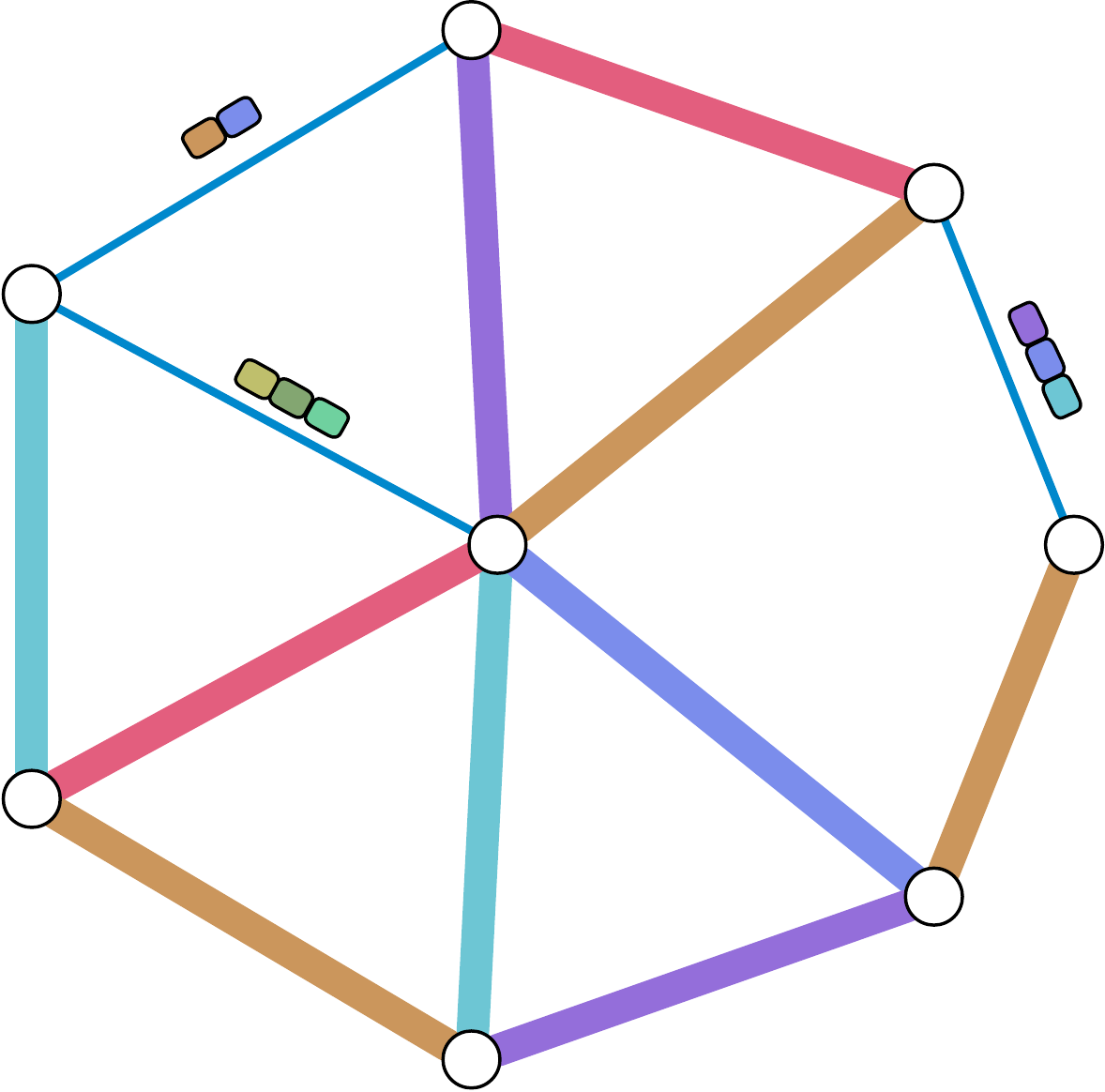}
	\end{minipage}%
	\begin{minipage}{.5\textwidth}
		\centering
		\includegraphics[width=.65\linewidth]{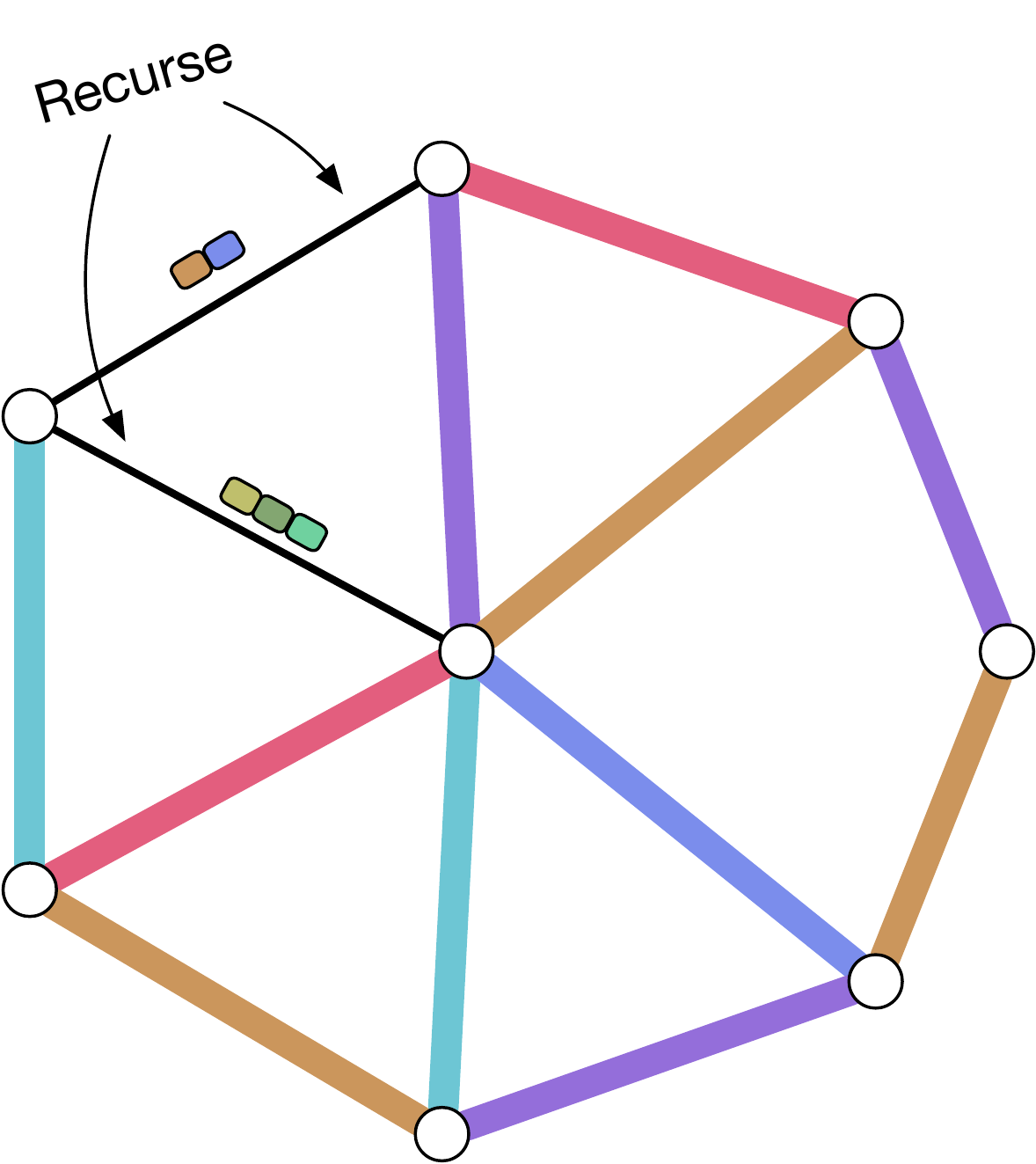}
	\end{minipage}
	\caption{Only one blue edge has a list of size strictly greater than $\deg(e)/2$, and thus is active in the coloring process. We repeat the whole reasoning on the remaining edges.}
	\label{fig:third-color}
\end{figure}

\paragraph{Defective edge coloring.}
We will now prove that it is possible to compute a $\frac{\deg(e)}{2 \beta}$-defective edge coloring with $O(\beta^2)$ colors in $O(\log^*X)$ rounds. The idea is the following: each node $v$ of degree $\deg(v)$ partitions its edges into $\lceil \frac{\deg(v)}{4 \beta}\rceil$ groups of size at most $4 \beta$. The edges in each group are then numbered with values $i$, where $1 \le i \le 4\beta$ (such that within each group, the edges are assigned unique values). For each edge, each node sends the assigned value over the edge. To each edge we assign a temporary color: the pair of values $(i,j)$ that have been transmitted over an edge, such that $i \le j$. Notice that there are at most $2$ edges for each group that have the same temporary color. This can happen when a node sends $i$ over an edge $e$ and $j$ over another edge $e'$, and receives $i$ over $e'$ and $j$ over $e$ (assuming $i \le j$, both edges get the color $(i,j)$). Hence, edges that have the same color \emph{and} are incident to the same group form paths or cycles. We can $3$-color the edges of these paths and cycles independently in $O(\log^*X)$ rounds. The final color of each edge is a triple containing $i$, $j$, and the color obtained in its path/cycle. Note that the number of possible colors is $O(\beta^2)$. We need to give a bound on the defect of an edge. For each endpoint of the edge, the defect is given by the number of groups that the node created, minus $1$ (if there is only one group, then all edges have different color, and the defect is $0$). Thus, the defect of an edge $\{u,v\}$ is given by $\lceil \frac{\deg(u)}{4 \beta}\rceil - 1 + \lceil \frac{\deg(v)}{4 \beta}\rceil - 1 \le \frac{\deg(u) + \deg(v)-2}{4 \beta} \le \frac{\deg(e) }{2 \beta} $.

\paragraph{Enough slack}
We need to ensure that, for each \emph{active} edge $e$, the list of remaining colors, after having removed colors used by the neighbors, satisfies $|L_e| > \beta \deg'(e)$, where $\deg'(e)$ is the degree of $e$ in the graph induced by active edges having the same color. By definition of defective edge coloring, $\deg'(e) \le \frac{\deg(e)}{2 \beta}$. Also, since each edge is marked as \emph{active} only if its list has size at least $|L_e| > \deg(e) /2$, then $|L_e|  > \deg(e) /2 \ge 2\beta\cdot\deg'(e)/2 \ge \beta \deg'(e)$, as required.

\paragraph{Running time}
We need to prove that the graph induced by uncolored edges (the ones where $c(e) = \bot$) is small. Consider an uncolored edge that has \emph{not} been marked as active. Its original list had size $|L_e| \ge \deg(e) + 1$. After removing colors used by the neighbors this list had size at most $\deg(e)/2$ (otherwise it would have been marked as active). This implies that $\deg(e) + 1 - \deg(e) /2 $ colors have been removed, that in turn implies that edge $e$ has at least $\deg(e) + 1 - \deg(e) /2 $  neighboring edges that are already colored, or in other words at most $\deg(e) - (\deg(e) + 1 - \deg(e) /2) = \deg(e) /2 -1$ neighboring edges that are still uncolored. Thus, the maximum degree becomes $\Deltaline / 2$. Hence, we can recurse and color the remaining edges in time $T(\Deltaline / 2, 1, C)$.

We proved that $T(\Deltaline, 1, C) $ is at most  $O(\log^*X) + O(\beta^2)\cdot
T\left(\frac{\Deltaline}{2\beta}, \beta, C\right) + T\left(\frac{\Deltaline}{2}, 1, C\right) $, where the first part is required to compute the defective edge coloring, the second part is required to solve $P(\frac{\Deltaline}{2\beta}, \beta, C)$ on each subgraph, and the last part is required to recurse on the remaining edges. By applying the same reasoning $O(\log \Deltaline)$ times we get the following:
\begin{align*}
T(\Deltaline, 1, C)  &\le \sum_{i=1}^{O(\log \Deltaline)} \left( O(\log^*X) + O(\beta^2) \cdot T\left(\frac{\Deltaline}{2\beta}, \beta, C\right) \right) + O(\log^* X)\\
&\le \sum_{i=1}^{O(\log \Deltaline)} \left( O(\log^*X) + O(\beta^2) \cdot T\left(\Deltaline, \beta, C\right) \right) + O(\log^* X)\\
&\le O(\beta^2\cdot\log \Deltaline)\cdot T(\Deltaline, \beta, C) + O(\log \Deltaline \log^*X).
\end{align*}

\subsection{Proof of \Cref{lemma:colorspacereduction}}\label{subsec:colorspacereduction} 
In this section we will show how to decompose the problem $P(\Deltaline,S,C)$ into many easier problems, and for this we will make use of the concept of list color space reduction, introduced in \cite{soda20_coloring}. Assume that we are given a list coloring instance with
lists consisting of colors from a color space of size $C$, and assume
that we are given a parameter $p$ and a partition of the color space
into $q=O(p)$ subspaces of size at most $C/p$. In a list color space
reduction, to each of the nodes (or edges in the case of edge coloring)
is assigned one of the $q$ subspaces of the overall color space. The new
list of the node is then given by the intersection of the old list with the
chosen color space. This divides the list coloring problem into
$q$ independent problems, each with a color space of size at most
$C/p$. These $q$ problems can then be solved in parallel, by applying the same algorithm on each subgraph induced by edges that chose the same color space. In the new
list coloring problems, the list size and degree of a node might
become smaller and the main objective is to make sure that the degree
of a node does not decrease at a much slower rate than the list
size. Hence, we need to reduce the degree of the subgraphs as much as possible, while trying to keep the size of the new lists as large as possible.

\begin{figure}
	\centering
	\includegraphics[width=0.9\textwidth]{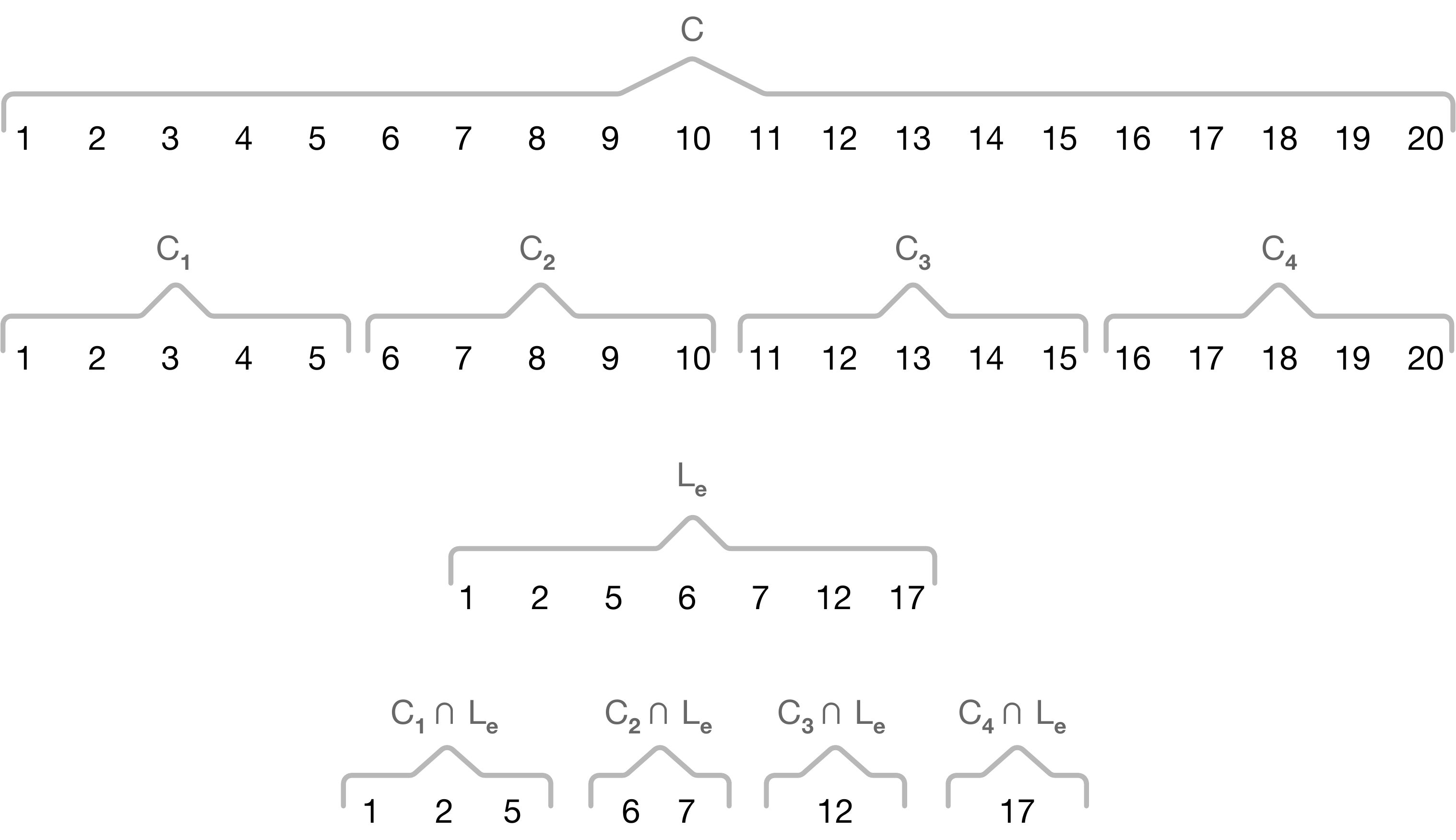}
	\caption{Example of list partitioning with $C=20$ and $p=4$, where the list given to an edge $e$ has size $7$. The indices of lists with large intersections are $I=\{1,2\}$, since $|C_1|,|C_1|\ge 2 \ge \frac{7}{2 \cdot H_4}$.}
	\label{fig:partition}
\end{figure}

We will start by proving a technical lemma, that shows that no matter how the original list instance looks like, we can always find enough subspaces that have large enough intersection with the original list. An example of an application of this lemma is shown in Figure \ref{fig:partition}.
\begin{lemma}\label{lemma:largeintersection}
	Assume that we have a list edge coloring instance for a graph
	$G=(V,E)$ where all list
	colors are from a color space $\calC$ and assume that we are given a
	partition of $\calC$ into $p\geq 1$ parts
	$\calC_1,\dotsc,\calC_p$. Then, for every edge $e\in E$ there exists
	an integer $k\in\set{1,\dotsc,p}$ for which there are $k$ indices
	$I\subseteq \set{1,\dotsc,p}$, $|I|=k$ such that for all $j\in I$,
	\[
	|L_e \cap C_j| \geq \frac{|L_e|}{k\cdot H_p}.
	\]
\end{lemma}
\begin{proof}
	For the sake of contradiction, assume that the claim of the lemma is
	not true for some edge $e\in E$. For simplicity, assume that the
	color subspaces $\calC_i$ are sorted by decreasing cardinality of
	$L_e\cap \calC_i$, i.e., assume that $|L_e\cap \calC_1|\geq
	|L_e\cap\calC_2|\geq \dotsc\geq |L_e\cap\calC_p|$. We then have 
	\begin{equation}\label{eq:smallsublists}
	\forall i\in \set{1,\dotsc,p}\,:\,|L_e\cap\calC_i|
	<\frac{|L_e|}{i\cdot H_p}.
	\end{equation}
	Hence, we obtain
	\[
	|L_e| = \sum_{i=1}^p |L_e\cap \calC_i| 
	\stackrel{\eqref{eq:smallsublists}}{<}
	\sum_{i=1}^p\frac{|L_e|}{i\cdot H_p} = |L_e|,
	\]
	which is a contradiction.
	
	This implies that \Cref{eq:smallsublists} is not true, and thus that there is a smallest
	$k\in\set{1,\dotsc,p}$ for which $|L_e\cap\calC_k|
	\geq|L_e|/(k\cdot H_p)$, and since for $i<k$, we have $|L_e\cap
	\calC_i|\geq |L_e\cap \calC_k|$, the claim of the lemma holds for
	an index set of size $k$. 
\end{proof}

\begin{figure}
	\centering
	\includegraphics[width=\textwidth]{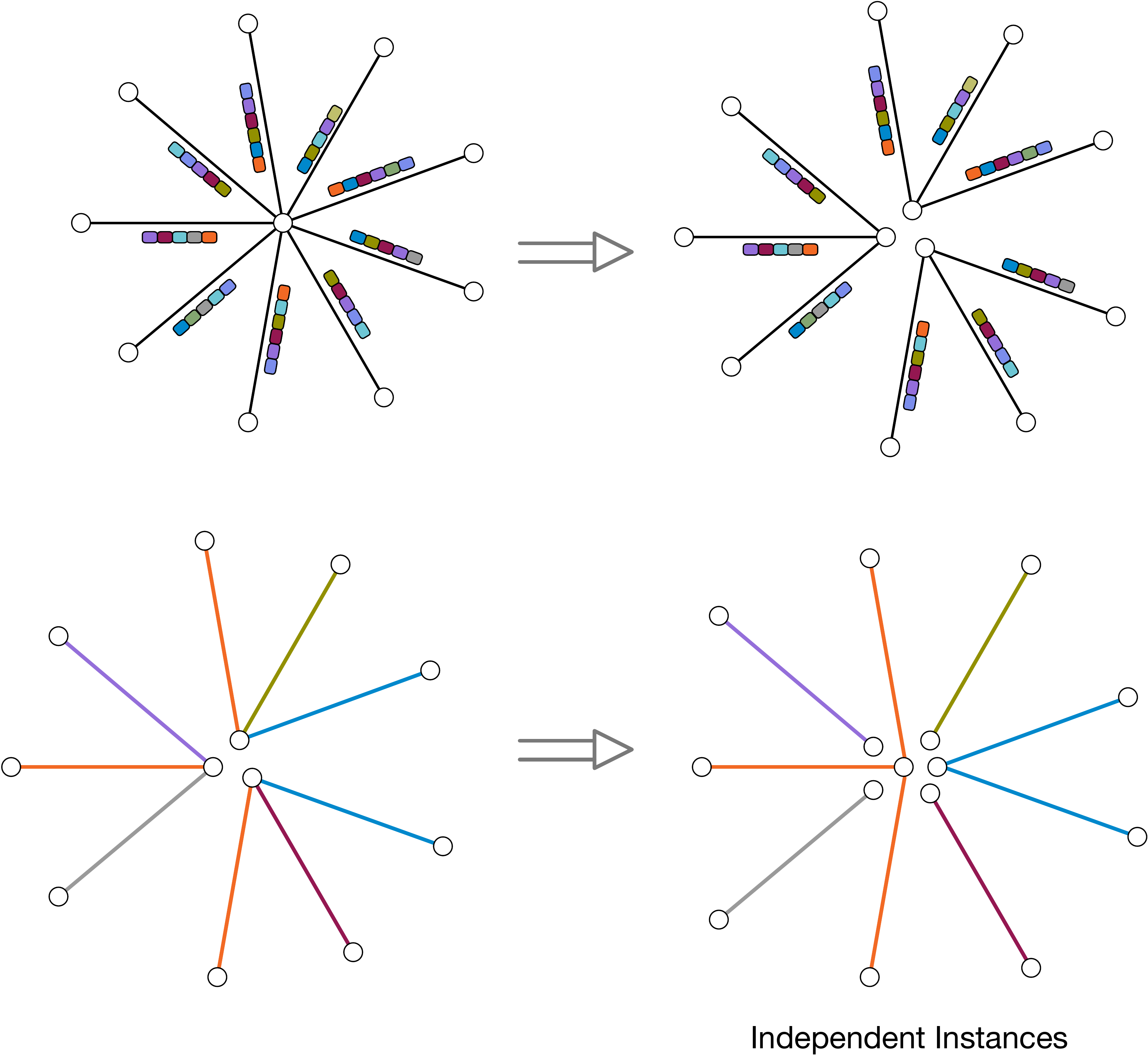}
	\caption{The list of each edge represents subspaces with large enough intersection with the original list. We split real nodes into multiple virtual nodes of smaller degree to obtain feasible list coloring instances.}
	\label{fig:decomposition}
\end{figure}

In the rest of the section, we will use the result of \Cref{lemma:largeintersection} to prove \Cref{lemma:colorspacereduction}. 
    Assume that we are given a list edge coloring instance of a graph
	$G=(V,E)$, where each list consists of colors from a color space
	$\calC$ of size $C:=|\calC|$. Assume further that we are given an
	integer parameter $p\in[2,C]$ and consider an arbitrary partition of
	the color space $\calC$ into $q\leq 2p$ subspaces
	$\calC_1,\dotsc,\calC_q$ of size at most $C/p$ (note that such a
	partition always exists). Assume that each edge $e\in E$ has a list
	$L_e$ of size larger than $(24\cdot H_{2p}\cdot \log p)\cdot \deg(e)$.
	In order to reduce the color space from size $C$ to at most $C/p$ we
	assign one of the $q$ subspaces $\calC_1,\dotsc,\calC_q$ to each edge
	$e\in E$. This generates $q$ independent list edge coloring
	instances that can then be solved in parallel. More precisely, for
	each $e\in E$, we
	define $i_e\in\set{1,\dotsc,q}$ to be (index of the) color subspace
	assigned to edge $e$. When
	assigning color subspace $\calC_{i_e}$ to $e$, the list $L_e$
	of $e$ is reduced to $L_e':=L_e\cap \calC_{i_e}$ and the degree of $e$
	is reduced to $\deg'(e)$, where $\deg'(e)$ is the number of
	neighboring edges $f\in E$ with $i_f=i_e$. We will show how to assign a color subspace
	to each edge such that for every edge $e$,
	\begin{equation}\label{eq:colorspacereduction}
	\deg'(e) \leq 24\cdot H_q\cdot \log p\cdot
	\frac{|L_e'|}{|L_e|}\cdot \deg(e)
	\leq 24\cdot H_{2p}\cdot \log p\cdot
	\frac{|L_e'|}{|L_e|}\cdot \deg(e).
	\end{equation}
	Note that the remaining $q$ independent list edge coloring
	instances can then be solved in time $T(\Deltaline, S/(24H_{2p}\log
	p),C/p)$, since \Cref{eq:colorspacereduction} implies that, compared to the ratio between the new and the old degrees, the ratio between the new and the old lists decreased by at most a $(24 \cdot H_{2p} \cdot \log
	p)$ factor. To prove the lemma, we need to show that the
	color subspaces can be assigned in time $(\log p)\cdot (1 + T(2p-1, 1, 2p))$.
	
	By \Cref{lemma:largeintersection}, for each edge $e\in E$, there is
	an integer $\ell(e)\in \set{0,\dotsc,\lfloor \log_2 q\rfloor}$ such
	that there exists a set $I\subseteq \set{1,\dotsc,q}$ of size
	$|I|\geq 2^{\ell(e)}$, where for each $i\in I$,
	$|L_e\cap\calC_i|\geq |L_e|/(2^{\ell(e)+1}\cdot H_q)$. In the
	following, we call $\ell(e)$ the level of edge $e$.
	
	If an edge $e$ is in a level $\ell(e)\leq 3$, we can assign the
	color subspace $\calC_i$ to $e$ that has the largest intersection
	with $L_e$. Note that, for $\ell(e)\leq 3$, by definition,
	$|L_e\cap \calC_i|\geq |L_e|/(16H_q)$, and thus the required
	condition of \Cref{eq:colorspacereduction} is satisfied even if all neighboring edges $f$ of $e$
	choose the same color subspace. We partition into two sets the remaining edges (i.e., the edges
	$e$ with $\ell(e)>3$): $E^{(1)}$ is the
	set of edges with $\ell(e)>3$ and for which
	$\deg(e)\geq 2^{\ell(e)}$, whereas $E^{(2)}$ is the set of edges
	with $\ell(e)>3$ and for which $\deg(e)< 2^{\ell(e)}$. We further partition the edges in
	$E^{(1)}$ according to their levels, and for
	each $\ell\in\set{4,\dotsc,\lfloor\log q\rfloor}$, we use
	$E^{(1)}_\ell$ to denote the edges $e\in E^{(1)}$ with
	$\ell(e)=\ell$. In the following, we first assign a color subspace
	$i_e$ to each edge $e\in E^{(1)}$ and we afterwards assign a color
	subspace $i_e$ to the edges in $E^{(2)}$.
	
	The edges in $E^{(1)}$ are assigned in $O(\log p)$ phases
	$4,\dotsc, \lfloor\log q\rfloor$. In phase
	$\ell\in \set{4,\dotsc, \lfloor \log p\rfloor}$, all edges
	$e\in E^{(1)}_\ell$ are active and choose a color subspace
	$i_e\in\set{1,\dotsc,q}$. At the beginning of phase $\ell$,
	for each edge $e\in E^{(1)}_\ell$, we determine the set
	$J_e\in\set{1,\dotsc,q}$ such that (I) for all $j\in J_e$,
	$|L_e\cap\calC_j|\geq |L_e|/(2^{\ell+1}\cdot H_q)$ and (II) the
	total number of neighboring edges of $e$ that have already
	chosen color subspace $\calC_j$ in an earlier phase is at most
	$\deg(e)/2^{\ell-1}$. We know that there are at least $2^\ell$ indexes $j$, where
	$j\in\set{1,\dotsc,q}$, for which
	$|L_e\cap\calC_j|\geq |L_e|/(2^{\ell+1}\cdot H_q)$. For at most
	$2^{\ell-1}$ of them, there can be $\deg(e)/2^{\ell-1}$ or more
	neighboring edges that have already chosen subspace $\calC_j$ and we
	thus have $|J_e|\geq 2^{\ell-1}$. For each node $v\in V$, we then
	divide its edges in $E^{(1)}_\ell$ into at most
	$\lceil\deg(v)/2^{\ell-2}\rceil$ groups of size at most
	$2^{\ell-2}$. For each of the groups, we create a virtual copy of
	$v$ so that the resulting virtual graph has a maximum degree of
	$2^{\ell-2}$. See Figure \ref{fig:decomposition} for an example. Note that the line graph of the virtual graph thus has
	a maximum degree of at most $2^{\ell-1}-2$. Our goal now is to
	assign to each edge $e\in E^{(1)}_\ell$ one of the subspaces in
	$J_e$ such that the edges that belong to the same virtual copy of
	some node $v$ are assigned different color subspaces. Since the line
	graph of the virtual graph has maximum degree at most $2^{\ell-1}-2$
	and each edge $e\in E_\ell$ has a set of $|J_e|> 2^{\ell-1} - 1$
	available color subspaces to choose from, the assignment of color
	subspaces is a $(\deg(e)+1)$-list edge coloring problem in the
	virtual graph. The color space of this list edge coloring problem is
	$\set{1,\dotsc,q}\subseteq\set{1,\dotsc,2p}$.  We can therefore solve this
	list edge coloring problem recursively in time
	$T(2^{\ell-1},1,q)\leq T(2p-1,1,2p)$. The time required to run a single
	phase is thus $1$ round for determining the set $J_e$ and at most
	$T(2p-1,1,2p)$ rounds for assigning a subspace $j\in J_e$.
	
	To conclude the discussion of the edges in $E^{(1)}$, we show that
	at the end of the above process, each edge $e\in E^{(1)}$ has at
	most $\deg(e) / 2^{\ell-2}$ neighboring edges $f$ for which $i_e=i_f$. In phase
	$\ell$, an edge of level $\ell$ can only choose color subspace
	$\calC_j$ if there are at most $\deg(e)/2^{\ell-1}$ neighboring
	edges that have already picked $\calC_j$ in previous phases. In the
	following, consider an edge $e=\set{u,v}$ for two nodes $u,v\in V$.
	In phase $\ell$, each node $u\in V$ assigns the same color subspace
	$\calC_j$ to at most $\lceil\deg(u)/2^{\ell-2}\rceil$ of its
	edges. Edge $e$ therefore has at most
	\[
	\left\lceil\frac{\deg(u)}{2^{\ell-2}}\right\rceil-1 +
	\left\lceil\frac{\deg(v)}{2^{\ell-2}}\right\rceil-1 \leq
	\frac{\deg(u)}{2^{\ell-2}} - \frac{1}{2^{\ell-2}} +
	\frac{\deg(v)}{2^{\ell-2}} - \frac{1}{2^{\ell-2}} =
	\frac{\deg(e)}{2^{\ell-2}}
	\]
	neighboring edges that get assigned the same color subspace in phase
	$\ell$ and thus at the end of phase $\ell$, $e$ has at most
	$\deg(e)/2^{\ell-1}+\deg(e)/2^{\ell-2}$ neighboring edges with the
	same color subspace $\calC_{i_e}$.
	
	To conclude the discussion of edges in $E^{(1)}_{\ell}$ for some
	given level $\ell$, we need to bound
	the number of neighboring edges that choose the same color subspace
	in later phases. Consider again edge $e=\set{u,v}\in
	E^{(1)}_\ell$ and let $D_e$ be the number of neighboring edges of
	$e$ that choose color space $i_e$ in phases $\ell'>\ell$. In each phase $\ell'>\ell$,
	every node $w\in V$ assigns at most
	$\lceil(\deg(u)-x_w)/2^{\ell'-2}\rceil$ edges to each color
	subspace, where $x_w$ denotes the number of edges of $w$ that were
	assigned a color subspace before phase $\ell'$. Since $e$ is
	assigned a color subspace before in phase $\ell<\ell'$, we clearly have
	$x_u,x_v\geq 1$. The value of $D_e$ can thus be bounded as
	\begin{eqnarray*}
		D_e 
		& \leq &
		\sum_{\ell'=\ell+1}^{\lfloor\log q\rfloor}
		\left(\left\lceil\frac{\deg(u)-1}{2^{\ell'-2}}\right\rceil +
		\left\lceil\frac{\deg(v)-1}{2^{\ell'-2}}\right\rceil\right)\\
		& < &
		\sum_{\ell'=\ell+1}^{\lfloor\log q\rfloor}
		\left(\frac{\deg(u)+\deg(v)-2}{2^{\ell'-2}}+2\right)\\
		& \leq &
		\sum_{\ell'=\ell+1}^{\lfloor\log q\rfloor}
		\left(\frac{\deg(e)}{2^{\ell'-2}} + 2\right)\\
		& < &
		\frac{\deg(e)}{2^{\ell-2}} + 2\log p.
	\end{eqnarray*}
	The last inequality follows because the number of phases
	$\ell'>\ell\geq 4$ is at most $\lfloor\log q\rfloor - \ell\leq
	\log(2p)-4 < \log p$. Overall, the number of neighboring edges of
	$e$ that choose color subspace $i_e$ is therefore at most
	\[
	\deg'(e)\leq \frac{\deg(e)}{2^{\ell-1}} + \frac{2\deg(e)}{2^{\ell-2}} + 2\log p
	\leq
	(10 + 2\log p)\cdot \frac{\deg(e)}{2^{\ell}}
	\leq (12\log p)\cdot \frac{\deg(e)}{2^{\ell}}.
	\]
	The inequality follows because for edges in $e\in E^{(1)}$, we have
	$\deg(e)\geq 2^{\ell(e)}$ and because $p\geq 2$.  By the choice of
	the level $\ell$ of $e$, we have
	$|L_e'|=|L_e\cap \calC_{i_e}|\geq |L_e|/(2^{\ell+1}H_{q})$ and
	\Cref{eq:colorspacereduction} therefore holds after assigning a
	color subspace to all edges in $E^{(1)}$.
	
	To show \Cref{eq:colorspacereduction} for all edges, it remains to
	assign a color subspace to each edge $e\in E^{(2)}$. Recall that for
	each $e\in E^{(2)}$, we have $\deg(e)<2^{\ell(e)}$. By the
	definition of the level $\ell(e)$, there are at least $2^{\ell}$
	color subspaces $i$ for which
	$|L_e\cap \calC_{i}|\geq |L_e|/(2^{\ell(e)+1}\cdot H_q)$ and thus in
	particular $L_e\cap \calC_{i}\neq 0$. Hence, edge $e$ can choose
	among more than $\deg(e)$ non-empty color subspaces. We can
	therefore assign a color subspace to $e$ such that no neighboring
	edge of $e$ is assigned the same color subspace. The assignment of
	color subspaces to the edges in $E^{(2)}$ is done by solving a
	$(\deg(e)+1)$-list edge coloring problem in a graph with maximum
	line graph degree at most $q-1\leq 2p-1$ and a color space of size
	at most $q\leq 2p$. For every edge $e\in E^{(2)}$, we get
	$\deg'(e)=0$ and thus \Cref{eq:colorspacereduction} is trivially
	satisfied for the edges in $E^{(2)}$. \Cref{eq:colorspacereduction}
	remains satisfied for edges in $E^{(1)}$ because the assignment to
	color subspaces of edges in $E^{(2)}$ does not add any new
	conflicts. The time for assigning the color subspaces to the edges
	in $E^{(2)}$ is given by the time to learn the subspaces that are still free
	(i.e., not assigned to a neighboring edge in $E^{(1)}$), that can be done in one round, and the time $T(2p-1,1,2p)$ to solve the resulting $(deg(e)+1)$-list edge
	coloring problem. The overall time for assigning the color subspaces
	to all edges is at most $(\log p)\cdot(1+T(2p-1,1,2p))$, since
	there are at most $\log2p-3 = \log p -2$ phases. This concludes the
	proof.

\subsection{Putting things together}\label{subsec:theorem}
In this section, we will prove \Cref{thm:listedgecoloring}. Informally, by \Cref{lemma:hardtoeasy} we can express $T(\cdot, 1, \cdot)$ as a function of $T(\cdot, \beta, \cdot)$, while by \Cref{lemma:colorspacereduction} we can express $T(\cdot, \beta, \cdot)$ as a function of both $T(\cdot, 1, \cdot)$ and  $T(\cdot, \beta', \cdot)$, for some $\beta' < \beta$. We will start by applying \Cref{lemma:colorspacereduction} multiple times, obtaining $T(\cdot, \beta, \cdot)$ only as a function of $T(\cdot, 1, \cdot)$. Then, we will combine \Cref{lemma:hardtoeasy}  with this new result, to obtain the theorem.

By applying \Cref{lemma:colorspacereduction} multiple times, we get the following.
\begin{lemma}\label{lemma:multiplereduction}
		Let $p\in[2, C]$ be an integer parameter. Let $k = \log_p C$. If $\Deltaline\geq 1$, $C\geq 2$,
		$S\geq (24\cdot H_{2p}\cdot \log p)^k$, and if an initial edge coloring with $X$ colors is given, we can express
		$T(\Deltaline, S, C)$ as
		\[
		T(\Deltaline, S, C) \leq (k \log p)\cdot(1+T(2p-1, 1, 2p)) + O(\log^* X).
		\]
\end{lemma}
\begin{proof}
	Given the assumption on $S$, it is possible to recursively apply \Cref{lemma:colorspacereduction} for $k$ times, without violating the requirements for $S$. At step $i$, we recurse on $T\left(\Deltaline, \frac{S}{(24\cdot H_{2p}\cdot \log p)^i}, \frac{C}{p^i}\right)$. Hence, when $i=k = \log_p C$, the palette size becomes constant, giving an instance that can be solved in $O(\log^*X)$ rounds. 
\end{proof}

We are now ready to prove \Cref{thm:listedgecoloring}. Let us assume we are given an edge list coloring instance that uses a color palette of size $\Deltaline^c$, for some constant $c \ge 1$. We start by computing an $O(\Deltaline^2)$-edge coloring in $O(\log^*n)$ rounds~\cite{linial87}. Then, we apply \Cref{lemma:hardtoeasy} with parameter $\beta = \alpha \log^{4c} \Deltaline$, for some large enough constant $\alpha$, and obtain that
\[
T(\Deltaline, 1, \Deltaline^c) \leq O(\log^{8c+1} \Deltaline)\cdot T(\Deltaline, \alpha \log^{4c} \Deltaline,  \Deltaline^c).
\]
We now set $p = \sqrt{\Deltaline}$. We get that $k = \log_p \Delta^c = 2c$, and that $(24 \cdot H_{2p} \cdot \log p)^k = O(\log^{4c} \Deltaline)$. Hence, we can apply \Cref{lemma:multiplereduction} and get the following:
\[
T(\Deltaline, 1, \Deltaline^c) \leq O(\log^{8c+2} \Deltaline)\cdot  T(2\sqrt{\Deltaline}-1,1,2\sqrt{\Deltaline})  + O(\log^{8c+2} \Deltaline).
\]
In other words, we get a polynomial reduction on the maximum degree, and by iterating the same reasoning $O(\log \log \Deltaline)$ times, the claim follows.

\section{Discussion} 
In their book on distributed graph coloring \cite{barenboimelkin_book}, Barenboim and Elkin asked if it was possible to solve $(2\Delta - 1)$-edge coloring in time polylogarithmic in the number of nodes. Since 2017, we know that this is possible \cite{FischerGK17}, and in this paper, we make a step forward towards better understanding the time complexity of this problem. We show that $(2\Delta - 1)$-edge coloring can be solved in time quasi-polylogarithmic in $\Delta$. This improves the best known upper bound for a large range of values of $\Delta$. Our result directly suggests the question whether it is possible to solve this problem in $O(\polylog\Delta + \log^* n)$ or even in $O(\log\Delta + \log^* n)$ deterministic rounds.


\bibliographystyle{alphaabbr}
\bibliography{references}

\end{document}